\documentclass[a4paper,10pt]{article}
\usepackage[utf8x]{inputenc}

\usepackage{amsmath,amssymb,amsfonts,amsthm,dsfont}

\usepackage{bbm}
\usepackage{hyperref}

\title{Minimax principles, Hardy-Dirac inequalities and operator cores for two and three dimensional Coulomb-Dirac operators}
\author{David M\"uller\footnote{Mathematisches Institut, Ludwig-Maximilians-Universit\"at M\"unchen, Theresienstr. 39, 80333 Munich, Germany \newline $dmueller@math.lmu.de$}}
\date{}
\newtheorem{thm}{Theorem}
\newtheorem{lem}[thm]{Lemma}
\newtheorem{cor}[thm]{Corollary}
\newtheorem{rem}[thm]{Remark}
\DeclareMathOperator{\Span}{span}
\DeclareMathOperator{\diag}{diag}
\begin{document}

\maketitle

\begin{abstract}
For $n\in\{2,3\}$ we prove minimax characterisations of eigenvalues in the gap of the $n$ dimensional Dirac operator with an potential, which may have a Coulomb singularity with a coupling constant up to the critical value $1/(4-n)$.  This result implies a so-called Hardy-Dirac inequality, which can be used to define a distinguished self-adjoint extension of the Coulomb-Dirac operator defined on  $\mathsf{C}_{0}^{\infty}(\mathbb{R}^n\setminus\{0\};\mathbb{C}^{2(n-1)})$, as long as the coupling constant does not exceed $1/(4-n)$. We also find an explicit description of an operator core of this operator.
\end{abstract}

\section{Introduction}
Throughout the text we assume that $n\in \{2,3\}$. In relativistic quantum mechanics an electron is described in $n$ dimensions by a $2(n-1)$ component spinor.  We say that a $2(n-1)\times 2(n-1)$ hermitian matrix function  $V$ on $\mathbb{R}^n$ is in $\mathfrak{P}_{n}$ if for some $\nu\in[0,1/(4-n))$ the inequality  $0\geq V\geq -\nu/|\cdot|\otimes \mathbb{I}_{\mathbb{C}^{2(n-1)}}$ holds and that $V$ belongs to $\overline{\mathfrak{P}}_{n}$ if  $0\geq V\geq -1/\big((4-n)|\cdot|\big)\otimes \mathbb{I}_{\mathbb{C}^{2(n-1)}}$ holds.\\
For $V\in\mathfrak{P}_n$ we denote by $D_n(V)$ the unique self-adjoint extension of 
\begin{align} \label{dirac_operator_symbol}
 \tilde{D}_n(V):=\begin{cases}&-\mathrm{i}\boldsymbol{\sigma}\cdot \nabla + \sigma_3 + V\text{ if }n=2 \\ &-\mathrm{i}\boldsymbol{\alpha}\cdot \nabla + \beta + V
 \text{ if }n=3
 \end{cases}
 \text{ defined on }\mathsf{C}_{0}^{\infty}(\mathbb{R}^n\setminus\{0\};\mathbb{C}^{2(n-1)})
,
\end{align}
with the property $\mathfrak{D}\big(D_n (V)\big)\subset \mathsf{H}^{1/2}(\mathbb{R}^n;
\mathbb{C}^{2(n-1)})$. The existence of this distinguished self-adjoint extension is proven in Section \ref{Section_Talman}. There 
we apply some general results developed in \cite{Nenciu}. In \eqref{dirac_operator_symbol} are $ \boldsymbol{\sigma}=(\sigma_1, \sigma_2),\boldsymbol{\alpha}=(\alpha_1,\alpha_2,\alpha_3)$  vectors; $\sigma_1,\sigma_2,\sigma_3$ the standard Pauli matrices; $\boldsymbol{\alpha}_{i}=\begin{pmatrix}
 0_{\mathbb{C}^2} & \sigma_{i} \\ \sigma_{i} & 0_{\mathbb{C}^2}
\end{pmatrix} $
for $i\in\{1,2,3\}$ and $\beta=\begin{pmatrix} \mathbb{I}_{\mathbb{C}^2} &  0_{\mathbb{C}^2} \\  0_{\mathbb{C}^2} &  - \mathbb{I}_{\mathbb{C}^2} \end{pmatrix} $. As in Proposition 1 in \cite{Cuenin-Siedentop} one can prove that there is a gap in  the essential spectrum of $D_n(V)$. To be more precise
\begin{align*}
                     \sigma_{\text{ess}}\big(D_n(V)\big)=(-\infty,-1]\cup[1,\infty).
\end{align*}  
In 1986 James D. Talman proposed in \cite{Talman} a formal minimax characterisation of the lowest eigenvalue in the gap of the essential spectrum of the operator $D_{3}(V)$. In this work we prove a minimax characterisation of eigenvalues in the gap of $D_{3}(V)$ in the spirit of Talman and an analogous result for $D_{2}(V)$. The exact result is:
\begin{thm}[Talman minimax principle]\label{Theorem-Talman-minimax-principle}
Let $V\in\mathfrak{P}_{n}$. If the $k^{\mathrm{th}}$ eigenvalue $\mu_k$ of $D_n(V)$ in $(-1,1)$, counted from below with multiplicity, exists, then it is given by 
 \begin{equation*}
 \mu_k =\inf\limits_{\substack{\mathfrak{M}\subset \mathsf{H}^{1/2}(\mathbb{R}^n;\mathbb{C}^{n-1}) \\ \dim \mathfrak{M}=k }}
 \sup\limits_{\psi \in (\mathfrak{M}\oplus  \mathsf{H}^{1/2}(\mathbb{R}^n;\mathbb{C}^{n-1}))\setminus\{0\}}
 \frac{\mathtt{d}_{n}[\psi]+\mathtt{v}[\psi]}{\|\psi\|^2}.
\end{equation*}  
Here $\mathtt{d}_n$ and $\mathtt{v}$ are the quadratic forms associated to the operators $D_{n}(0)$ and $V$. 
\end{thm}
About Theorem \ref{Theorem-Talman-minimax-principle} we want to remark that for $n=3$ 
there is an historical overview of results of the same type in \cite{Morozov-Mueller} 
and that for $n=2$ there is no comparable result known. Moreover, Theorem \ref{Theorem-Talman-minimax-principle} improves in the three dimensional case Theorem 3 in 
\cite{Morozov-Mueller}, which is the best known result  for a Dirac operator with an electrostatic potential having strong Coulomb singularity.\\
Furthermore, we give a different proof of the Esteban-Séré minimax principle (see Theorem 2 in \cite{Morozov-Mueller} and \cite{Esteban-Sere}) and prove an analogous result for two dimensional Dirac operators:
\begin{thm}[Esteban-Séré minimax principle] \label{Theorem_Esteban_Sere}
Let $V\in\mathfrak{P}_{n}$. If the $k^{\mathrm{th}}$ eigenvalue $\mu_k$ of $D_n(V)$ in $(-1,1)$, counted from below with multiplicity, exists, then it is given by
\begin{align*} 
 \mu_k =\inf\limits_{\substack{\mathfrak{M}\subset P^{+}_n\mathsf{H}^{1/2}(\mathbb{R}^n;\mathbb{C}^{2(n-1)}) \\ \dim \mathfrak{M}=k }}
 \sup\limits_{\psi \in (\mathfrak{M}\oplus  P^{-}_n\mathsf{H}^{1/2}(\mathbb{R}^n;\mathbb{C}^{2(n-1)}))\setminus\{0\}}
 \frac{\mathtt{d}_{n}[\psi]+\mathtt{v}[\psi]}{\|\psi\|^2}.
\end{align*}
Here $P^{+}_n$ is the projector on the non-negative spectral subspace of $D_n(0)$ and $P^{-}_n:=\mathbb{I}-P^{+}_n$.
\end{thm}
A direct application of Theorem \ref{Theorem-Talman-minimax-principle} is:
\begin{thm}[Hardy-Dirac inequality]\label{Theorem-Hardy-Dirac}
Let $v$ be a scalar function on $\mathbb{R}^n$ such that $v\otimes \mathbb{I}_{\mathbb{C}^{2(n-1)}}\in \mathfrak{P}_{n}$. 
Moreover, we define the operator:
\begin{align*}
                 K_n:=\begin{cases} 
                  -\mathrm{i}\partial_1-\partial_2 \text{ if }n=2,\\ -\mathrm{i}\boldsymbol{\sigma}\cdot \nabla
                  \text{ if }n=3,
                 \end{cases}
\end{align*}
and denote by $\lambda(v)$ the smallest eigenvalue of $D_n(v\otimes \mathbb{I}_{\mathbb{C}^{2(n-1)}})$ in the gap $(-1,1)$.
Then for all $ \varphi\in \mathsf{H}^{1}(\mathbb{R}^n;\mathbb{C}^{n-1})$ the inequality
\begin{align}\label{Hardy_Dirac_Inequality}
0\leq
\int\limits_{\mathbb{R}^n}\frac{|K_n\varphi(\mathbf{x})|^2}{1+\lambda(v)-v(\mathbf{x})}\mathrm{d}\mathbf{x}+ \int\limits_{\mathbb{R}^n}\big(1-\lambda(v)+v(\mathbf{x})\big)|\varphi(\mathbf{x})|^2\mathrm{d}\mathbf{x}
\end{align}
holds.
\end{thm}
We follow the tradition of  \cite{Dolbeault-Esteban-Loss-Vega}  and call these type of inequality Hardy-Dirac inequality. In \cite{Dolbeault-Esteban-Sere} it is demonstrated, how one can prove Hardy-Dirac inequalities for $n=3$ with the help of the Talman minimax principle.\\
We know that the lowest eigenvalue of $D_n(-\nu/|\cdot|\otimes \mathbb{I}_{\mathbb{C}^{2(n-1)}})$ in $(-1,1)$ is $\sqrt{1-\big((4-n)\nu\big)^2}$ for $\nu \in \big(0,1/(4-n)\big)$ (see \cite{Dong-Zong} and \cite{Thaller}). Thus Theorem \ref{Theorem-Hardy-Dirac} implies with a simple limiting argument
\begin{cor} \label{Corollary_Dirac_Coulomb_Hardy}
Let $\nu \in [0,1/(4-n)]$. Then 
\begin{align*}
0\leq\int\limits_{\mathbb{R}^n} \Bigg(\frac{|K_n\varphi |^2}{1+
\sqrt{1-\big((4-n)\nu\big)^2}+\frac{\nu}{|x|}}+
\bigg(1-\sqrt{1-\big((4-n)\nu\big)^2}-\frac{\nu}{|x|}\bigg)|\varphi|^2\Bigg) \mathrm{d}\mathbf{x}
\end{align*}
holds for all $ \varphi\in\mathsf{H}^{1}(\mathbb{R}^n;\mathbb{C}^{n-1})$.
\end{cor}
Let $\nu\in \big[ 0,1/(4-n)\big]$. 
With the help of Corollary \ref{Corollary_Dirac_Coulomb_Hardy} and Theorem 1 in \cite{Esteban-Loss-abstract} ($\tilde{D}_n(-\nu/|\cdot|\otimes \mathbb{I}_{\mathbb{C}^{2(n-1)}})$ corresponds to $H$ there) we know that there is only one self-adjoint extension of $\tilde{D}_n(-\nu/|\cdot|\otimes \mathbb{I}_{\mathbb{C}^{2(n-1)}})$ with a positive Schur complement. We denote this distinguished self-adjoint extension by $D^{\nu}_n$. Now we want to give an explicit description of an operator core of $D^{\nu}_n$. For this purpose we introduce polar and spherical coordinates. We denote by the coordinate pair $(\rho,\vartheta)\in [0,\infty)\times [0,2\pi)$ the radial and angular polar coordinates in $\mathbb{R}^2$ and by the coordinate triplet $(r,\theta,\phi)\in [0,\infty)\times [0,\pi)\times [0,2\pi)$ the radial, inclination and azimuthal spherical coordinates in $\mathbb{R}^3$. For $m\in \{-1/2,1/2\}^{n-1}$ we define the function $\zeta_{n,m}^{\nu}$ in polar coordinates for $n=2$
\begin{align}\label{special_function_2d}
             \zeta_{2,m}^{\nu}(\rho,\vartheta):= 
             \xi(\rho)\rho^{\sqrt{1/4-\nu^2}-1/2}
             \begin{pmatrix} 
             \nu \frac{\mathrm{e}^{-\mathrm{i}(1/2+m)\vartheta}}{\sqrt{2\pi}} \\
             -\mathrm{i}\big(\sqrt{1/4-\nu^2}+(-1)^{1/2-m}/2\big)   
             \frac{\mathrm{e}^{\mathrm{i}(1/2-m)\vartheta}}{\sqrt{2\pi}}
             \end{pmatrix}; 
\end{align} 
and in spherical coordinates for $n=3$
\begin{align} \label{special_function_3d}
             \zeta_{3,m}^{\nu}(r,\theta,\phi)&:= 
             \xi(r)r^{\sqrt{1-\nu^2}-1}
             \begin{pmatrix}
             \nu \Omega_{\frac{1}{2}+m_2,m_1,-m_2}
             (\theta,\phi)\\
             -\mathrm{i}\big(\sqrt{1-\nu^2}+(-1)^{\frac{1}{2}-m_2}\big)  \Omega_{\frac{1}{2}-m_2,m_1,m_2}(\theta,\phi)
             \end{pmatrix};  
\end{align} 
with the spherical spinor $\Omega_{l,m,s}$ (see Relation (7) in \cite{Evans-Perry-Siedentop}) and the smooth cut-off function $\xi$ (i.e., $\xi \in \mathsf{C}^{\infty}(\mathbb{R}_{+};\mathbb{R}_{+})$, $ \xi(t)=1 $ for $t\in (0,1)$ and $ \xi(t)=0 $ for $t>2$). 
In the next theorem we give a characterisation of an operator core of $D_n^{\nu}$ 
with the help of the functions $\zeta_{n,m}^{\nu}$ introduced in  \eqref{special_function_2d} and \eqref{special_function_3d}.
\begin{thm}[Operator core] \label{Theorem_operator_core}
Let $\nu\in \big[0,1/(4-n)\big]$. The set
\begin{align} \label{operator_core}
             \mathfrak{C}_{n}^{\nu}:= \mathsf{C}_{0}^{\infty}(\mathbb{R}^n\setminus\{0\};\mathbb{C}^{2(n-1)})\dot{+}
             \begin{cases} 
                     \{0\},\text{ if }n=2,\ \nu=0\text{ or }n=3,\ \nu \in \big[0,\frac{\sqrt{3}}{2}\big]; \\
                     \Span\{\zeta_{n,m}^{\nu} : m \in \{-1/2,1/2\}^{n-1}\},\text{ else};         
             \end{cases}
\end{align} 
is an operator core for $D_n^{\nu}$.
\end{thm}
The knowledge of the operator core of $D_n^{\nu}$ is important for the proof of estimates on the square of the operator, see e.g. \cite{Morozov-Mueller-CLR}. In Remark \ref{Remark_operator_core_nenciu} we show that for $\nu\in\big[0,1/(4-n)\big)$ the set $\mathfrak{C}_{n}^{\nu}$ is an operator core for  $D_n(-\nu/|\cdot|\otimes \mathbb{I}_{\mathbb{C}^{2(n-1)}})$. A direct consequence is:
\begin{cor}
Let $\nu\in [0,1/(4-n))$. The distinguished self-adjoint extensions of $\tilde{D}_n(-\nu/|\cdot|\otimes \mathbb{I}_{\mathbb{C}^{2(n-1)}})$ in the sense of \cite{Nenciu} and \cite{Esteban-Loss-abstract} coincide, i.e.,
\begin{align*}
                        D_n^{\nu}=D_n(-\nu/|\cdot|\otimes \mathbb{I}_{\mathbb{C}^{2(n-1)}}).  
\end{align*}
\end{cor}
The proofs of the minimax characterisations rely on the  angular momentum channel decomposition of the Coulomb-Dirac operator in the momentum space. This representation and the 
corresponding unitary transformations are introduced in the next section. In the remaining sections we prove in the order of enumeration: Theorems \ref{Theorem-Talman-minimax-principle}, \ref{Theorem_Esteban_Sere}, \ref{Theorem-Hardy-Dirac} and \ref{Theorem_operator_core}.
\section{Angular momentum channel decomposition in the momentum space}
The Fourier transform connects the quantum mechanical descriptions of a particle in the configuration and momentum
space. We use the standard unitary Fourier transform $\mathcal{F}_n$ in $\mathsf{L}^2(\mathbb{R}^n)$ given for $\varphi \in 
\mathsf{L}^1(\mathbb{R}^n)\cap\mathsf{L}^2(\mathbb{R}^n)$ by
\begin{align} \label{Definition_Fourier_transform}
              \mathcal{F}_{n}\varphi:=\frac{1}{(2\pi)^{n/2}}\int\limits_{\mathbb{R}^n}\mathrm{e}^{-\mathrm{i}\langle\cdot,\mathbf{x}\rangle} 
              \varphi(\mathbf{x}) 
               \mathrm{d}\mathbf{x}.
\end{align}
For the angular momentum channel decomposition in $n$ dimensions we use an orthonormal basis in $\mathsf{L}^2(\mathbb{S}^{n-1};\mathbb{C}^{n-1})$.  For $n=2$ this orthonormal basis is $\big((2\pi)^{-1/2}\mathrm{e}^{\mathrm{i}m (\cdot)}\big)_{m\in\mathbb{Z}}$. In three dimensions we use spherical spinors $\Omega_{l,m,s}$, which are defined in Relation (7) in \cite{Evans-Perry-Siedentop}, with $l\in\mathbb{N}_{0}$, $\ m\in \{-l-1/2,\ldots,l+1/2\}$ and $s\in\{-1/2,1/2\}$. The corresponding  index sets are denoted by 
\begin{align} \label{T_2}
           \mathfrak{T}_2:=\mathbb{Z};
\end{align}
and
\begin{align} \label{T_3}
\mathfrak{T}_{3}:=
                  \Bigg\{(l,m,s): l\in \mathbb{N}_{0}, m\in \bigg\{-l-\frac{1}{2},\ldots,l+
                  \frac{1}{2}\bigg\},s=\pm\frac{1}{2},\Omega_{l,m,s}\neq 0 \Bigg\}     .            
 \end{align}
Furthermore, we define subsets  $\mathfrak{T}^{\pm}_n$ of $\mathfrak{T}_{n} $:
\begin{align}\label{Tpm}
       \mathfrak{T}^{a}_n:=\begin{cases}
       2\mathbb{Z} &\text{ if }n=2,\ a=+;\\
       2\mathbb{Z}+1 &\text{ if }n=2,\ a=-;\\
       \{(l,m,s)\in \mathfrak{T}_3: s=\pm 1/2\}&\text{ if }n=3,\ a=\pm.
       \end{cases}
\end{align} 
Note that if $(l,m,-1/2)\in \mathfrak{T}^{-}_3$ then $l \in \mathbb{N}$. \\
Moreover, we introduce bijective maps 
\begin{align}
T_2: \mathfrak{T}_2&\rightarrow \mathfrak{T}_2,\ T_2 k:=k+1;
\end{align}
and
\begin{align}
T_3: \mathfrak{T}_3 \rightarrow \mathfrak{T}_3,\ T_3(l,m,s):=(l+2s,m,-s).
\end{align}
We can represent any $\varphi\in\mathsf{L}^2(\mathbb{R}^2;\mathbb{C})$ in polar coordinates and $\zeta\in\mathsf{L}^2(\mathbb{R}^3;\mathbb{C}^2)$ in spherical coordinates as 
\begin{align} 
         \varphi(\rho,\vartheta)&=\sum\limits_{k\in\mathfrak{T}_2}(2\pi \rho)^{-1/2} 
\varphi_{k}(\rho)\mathrm{e}^{\mathrm{i}k\vartheta}; \label{Definition_representation_2}\\
         \zeta(r,\theta,\phi)&=\sum\limits_{(l,m,s)\in \mathfrak{T}_3}r^{-1} \zeta_{(l,m,s)}(r) \Omega_{l,m,s}(\theta,\phi);
         \label{Definition_representation_3}
\end{align}
with 
\begin{align}
         \varphi_{k}(\rho)&:=\sqrt{\frac{\rho}{2\pi}}
         \int\limits_{0}^{2\pi}\varphi(\rho,\vartheta) \mathrm{e}^{-\mathrm{i}k\vartheta}
         \mathrm{d}\vartheta ;\label{Definition_representation_4}\\
         \zeta_{(l,m,s)}(r)&:=r\int\limits_{0}^{2\pi}\int\limits_{0}^{\pi}\big\langle \Omega_{l,m,s}
           (\theta,\phi),\zeta(r,\theta,\phi)\big\rangle_{\mathbb C^2} 
           \sin(\theta)\mathrm{d}\theta\mathrm{d}\phi. \label{Definition_representation_5}
\end{align}
With the help of  \eqref{Definition_representation_4} and \eqref{Definition_representation_5} we define the unitary 
operator
\begin{align} \label{Definition_U_n}
\mathcal{U}_n: \mathsf{L}^2(\mathbb{R}^n;\mathbb{C}^{n-1}) \rightarrow \bigoplus\limits_{j\in\mathfrak{T}_n}\mathsf{L}^2(\mathbb{R}_{+});\quad \psi \mapsto   \bigoplus\limits_{j\in\mathfrak{T}_n}\psi_j. 
\end{align}
For the proof of the following lemma see Theorem 2.2.5 in \cite{Balinsky-Evans} (based on Lemmata 2.1, 2.2 of \cite{Bouzouina}) for $n=2$ and Section 2.2 in \cite{Balinsky-Evans} for $n=3$.
\begin{lem}\label{Lemma_Coulomb_Channel_Decomposition}
For $j\in \big(\mathbb{N}_{0}/2-1/2\big)$ and $z\in (1, \infty)$ let
\begin{equation}\label{Q}
 Q_{j}(z)= 2^{-j-1}\int_{-1}^1(1- t^2)^{j}(z -t)^{-j-1}\,\mathrm dt
\end{equation}
be a Legendre function of the second kind (see Section 15.3 in \cite{Whittaker-Watson}).
Let the ses\-qui\-li\-ne\-ar form $\mathtt{q}_j$ be defined on 
$\mathsf L^2\big(\mathbb R_+, (1+ p^2)^{1/2}\mathrm dp\big)\times\mathsf L^2\big(\mathbb R_+, (1+ p^2)^{1/2}\mathrm dp\big)$ by
\begin{equation}\label{q_m}
 \mathtt{q}_j[f,g] := \pi^{-1}\int_{0}^{\infty}\int_{0}^{\infty}\overline{f(p)}Q_{j}\bigg(\frac12\Big(\frac {q}p +\frac p{q}\Big)\bigg)g(q)\,\mathrm dq\,\mathrm dp.
\end{equation}
For the special case $f=g$ we introduce $\mathtt{q}_j[f]:=\mathtt{q}_j[f,f]$.\\
Then for every $\zeta,\eta\in \mathsf H^{1/2}(\mathbb R^n)$ the relation
\begin{equation}\label{Coulomb form}
\int_{\mathbb{R}^n}\frac{\overline{\zeta}(\mathbf{x})\cdot \eta(\mathbf{x})}{|\mathbf{x}|}\mathrm{d}\mathbf{x}
 =\begin{cases}
  \sum\limits_{k\in \mathfrak{T}_2} \mathtt{q}_{|k|-1/2}\big[(\mathcal F_2 \zeta )_k,(\mathcal F_2 \eta )_k\big]\text{ if }n=2,\\
  \sum\limits_{(l,m,s)\in\mathfrak{T}_3} \mathtt{q}_{l}\big[
  (\mathcal F_3 \zeta )_{(l,m,s)},(\mathcal F_3 \eta)_{(l,m,s)}\big]\text{ if }n=3,
  \end{cases}
\end{equation}
holds. 
\end{lem}
The operators $-\mathrm{i}\boldsymbol{\sigma}\cdot\nabla$ and $-\mathrm{i}\boldsymbol{\alpha}\cdot\nabla$ are partially diagonalised in the momentum space by the unitary transforms 
\begin{align} \label{Definition_W_2}
         \mathcal{W}_2: \mathsf{L}^2(\mathbb{R}^2;\mathbb{C}^2) \rightarrow
          \bigoplus\limits_{k \in\mathfrak{T}_2}
         \mathsf{L}^2(\mathbb{R}_{+};\mathbb{C}^2);\quad \begin{pmatrix}
         \varphi \\ \psi         \end{pmatrix} \mapsto  
         \bigoplus\limits_{k \in\mathfrak{T}_2}
         \begin{pmatrix}
         \varphi_{k} \\ \psi_{T_2 k}         \end{pmatrix}
\end{align}
and
\begin{align} \label{Definition_W_Teil_1}
         \mathcal{W}_3: \mathsf{L}^2(\mathbb{R}^3;\mathbb{C}^4) \rightarrow \bigoplus\limits_{(l,m,s)\in\mathfrak{T}_3}
         \mathsf{L}^2(\mathbb{R}_{+};\mathbb{C}^2);\quad \begin{pmatrix}
         \psi_1 \\ \psi_2 \\ \psi_3 \\ \psi_4
         \end{pmatrix} \mapsto 
         \bigoplus_{(l,m,s)\in\mathfrak{T}_3}
         \begin{pmatrix}
                \psi_{(l,m,s)}^{+} \\ 
                 \psi_{T_3 (l,m,s)}^{-}
         \end{pmatrix}
\end{align}
with 
\begin{align}\label{Definition_W_Teil_2}
               \psi_{(l,m,s)}^{+}:=\begin{pmatrix}
                        \psi_1 \\ \psi_2 
                \end{pmatrix}_{(l,m,s)} \text { and } \psi_{(l,m,s)}^{-}:=\begin{pmatrix}
                        \psi_3\\ \psi_4
                \end{pmatrix}_{(l,m,s)}
\end{align}
for $(l,m,s)\in\mathfrak{T}_3$. To be more precise: 
\begin{lem} \label{Lemma_alpha_dot_impuls}
For the self-adjoint operators  $-\mathrm{i}\boldsymbol{\sigma}\cdot\nabla$ and $-\mathrm{i}\boldsymbol{\alpha}\cdot\nabla$ the relations
\begin{align}        \label{kinetic}         
             (\mathcal{W}_n\mathcal{F}_n)^{*} \Bigg( \bigoplus\limits_{j\in \mathfrak{T}_n}
             \begin{pmatrix}
             0 & (\cdot) \\
             (\cdot) & 0
             \end{pmatrix}\Bigg)\big(\mathcal{W}_n\mathcal{F}_n\big)
             =\begin{cases}
             -\mathrm{i}\boldsymbol{\sigma}\cdot \nabla\text{ if }n=2, \\
             -\mathrm{i}\boldsymbol{\alpha}\cdot \nabla\text{ if }n=3, 
             \end{cases}
\end{align}
hold.
\end{lem}
\begin{proof}
By a straightforward calculation and Relation 2.1.28 in \cite{Balinsky-Evans} the 
relations 
\begin{align}\label{sigma_dot_x}
\boldsymbol{\sigma}\cdot \mathbf{x} &= \begin{pmatrix}
    0 & \mathrm{e}^{-\mathrm{i}\vartheta}\rho \\ \mathrm{e}^{\mathrm{i}\vartheta}\rho & 0
\end{pmatrix}\text{ for }\mathbf{x}\in\mathbb{R}^2; \\
\label{sigma_dot_p_spherical_spinors}
         \boldsymbol{\sigma}\cdot \frac{\mathbf{x}}{|\mathbf{x}|}\Omega_{l,m,s}&=\Omega_{l+2s,m,-s}\text{ for }\mathbf{x}\in\mathbb{R}^3\text{ and }(l,m,s)\in\mathfrak{T}_3; 
\end{align}
hold.\\
The set $\mathsf{C}_{0}^{\infty}(\mathbb{R}^n;\mathbb{C}^{2(n-1)})$ is dense in 
$ \mathsf{H}^{1}(\mathbb{R}^n;\mathbb{C}^{2(n-1)})$. Thus it is enough to work with $\psi \in\mathsf{C}^{\infty}_{0}(\mathbb{R}^2;\mathbb{C}^{2})$ and $\zeta \in\mathsf{C}^{\infty}_{0}(\mathbb{R}^3;\mathbb{C}^{4})$.\\ 
Moreover, the Fourier transform diagonalises differential operators: 
\begin{align}
\langle \psi, -\mathrm{i}\boldsymbol{\sigma}\cdot \nabla \psi \rangle =\langle \mathcal{F}_2\psi,\boldsymbol{\sigma}\cdot \boldsymbol{p}\mathcal{F}_2\psi \rangle , \label{Proof_Sigma_P_0}\\
\langle \zeta, -\mathrm{i}\boldsymbol{\alpha}\cdot \nabla \zeta \rangle =\langle \mathcal{F}_3 \zeta,\boldsymbol{\alpha}\cdot \boldsymbol{p}\mathcal{F}_3 \zeta \rangle .   \label{Proof_Sigma_P_1}
\end{align}
Here we denote by $\boldsymbol{p}$ the independent variable of multiplication in $\mathsf{L}^{2}(\mathbb{R}^n;\mathrm{d}\mathbf{p})$. \\
Now we prove \eqref{kinetic} for $n=3$. We obtain by the representation \eqref{Definition_representation_3}  of the upper and lower bispinor of $\mathcal{F}_3\zeta$ and the notation introduced in \eqref{Definition_W_Teil_2} that 
the right hand side of \eqref{Proof_Sigma_P_1} is equal to
\begin{equation} \begin{split}
2\sum\limits_{\substack{(l',m',s')\in \mathfrak{T}_3\\ (l,m,s)\in \mathfrak{T}_3}}
  \Re\left(\big\langle |\boldsymbol{p}|^{-1}\big(\mathcal{F}_3\zeta\big)_{(l',m',s')}^{+}\Omega_{l',m',s'},(\boldsymbol{\sigma}\cdot \boldsymbol{p}) |\boldsymbol{p}|^{-1} \big(\mathcal{F}_3\zeta\big)_{(l,m,s)}^{-} \Omega_{l,m,s}\big\rangle\right)   . \end{split}   \label{Proof_Sigma_P_2}
\end{equation}
The expression in \eqref{Proof_Sigma_P_2} is equal to
\begin{align}
&2\sum\limits_{(l,m,s)\in \mathfrak{T}_3}
  \Re\left(\big\langle \big(\mathcal{F}_3\zeta\big)_{(l+2s,m,-s)}^{+},(\cdot)\big(\mathcal{F}_3\zeta\big)_{(l,m,s)}^{-}\big\rangle\right)      
  \nonumber \\
&=\sum\limits_{(l,m,s)\in \mathfrak{T}_3}
  \left\langle 
         \begin{pmatrix}
                \big(\mathcal{F}_3\zeta\big)_{(l,m,s)}^{+} \\ 
                \big(\mathcal{F}_3\zeta\big)_{T_3(l,m,s)}^{-}
         \end{pmatrix},
         \begin{pmatrix}
             0 & (\cdot ) \\
             (\cdot ) & 0
             \end{pmatrix}
          \begin{pmatrix}
                \big(\mathcal{F}_3\zeta\big)_{(l,m,s)}^{+} \\ 
                \big(\mathcal{F}_3\zeta\big)_{T_3(l,m,s)}^{-}
         \end{pmatrix}
  \right\rangle  \nonumber \\
&=\left\langle
   \mathcal{W}_3\mathcal{F}_3\zeta,\Bigg(\bigoplus_{(l,m,s)\in \mathfrak{T}_3}
             \begin{pmatrix}
             0 & (\cdot ) \\
             (\cdot ) & 0
             \end{pmatrix}\Bigg)
    \mathcal{W}_3\mathcal{F}_3\zeta 
  \right\rangle \label{Proof_Sigma_P_3}               
\end{align}
by the sequential application of  \eqref{sigma_dot_p_spherical_spinors}, \eqref{Definition_W_Teil_1} and \eqref{Definition_Fourier_transform}. Thus the claim of Lemma \ref{Lemma_alpha_dot_impuls} is a consequence of \eqref{Proof_Sigma_P_1}, \eqref{Proof_Sigma_P_2} and \eqref{Proof_Sigma_P_3}.\\
For $n=2$ we obtain \eqref{kinetic} by an analogous procedure, i.e., we 
represent the upper and lower component of $\mathcal{F}_2\psi $ by \eqref{Definition_representation_2} in \eqref{Proof_Sigma_P_0} and perform a calculation, which involves \eqref{sigma_dot_x}. 
\end{proof}

\section{Proof of Theorem \ref{Theorem-Talman-minimax-principle}}\label{Section_Talman}
Let $V\in\mathfrak{P}_{n}$. We use the abstract minimax principle Theorem 1 of \cite{Morozov-Mueller} to prove the Talman minimax principle. We apply the theorem with $q:=\mathtt{d}_{n}$ (quadratic form associated to 
$D_n(0)$), $B:=D_n(V)$ and $\Lambda_{\pm}$ as the projector $T^{\pm}_n$ on the upper and lower $(n-1)$ components of a $2(n-1)$ spinor, i.e., 
\begin{align*}
     T^{+}_n\binom{\varphi}{\psi}=\binom{\varphi}{0}, \qquad T^{-}_n\binom{\varphi}{\psi}=\binom{0}{\psi},\text{ for }\varphi,
     \psi \in \mathsf{L}^2(\mathbb{R}^n;\mathbb{C}^{n-1}).
\end{align*}  
That $D_n(V)$ plays the role of $B$ in \cite{Morozov-Mueller} is a consequence of Theorem 2.1 in \cite{Nenciu} and the following lemma.
\begin{lem}\label{Lemma_form_perturbation}
Let $V\in\mathfrak{P}_n$. Then the quadratic form $\mathtt{v}$ associated to the operator $V$ is a form perturbation of $D_n(0)$ in the sense of Definition 2.1 in \cite{Nenciu}.
\end{lem}
\begin{proof}
$V$ is $D_n(0)$ form bounded by the Herbst inequality (see Theorem 2.5 in \cite{Herbst}). Moreover, the inequality
\begin{align*}
                  \|r^{-1/2}D_n(0)^{-1}r^{-1/2}\|\leq 4-n
\end{align*}
holds. This is proven in Section 2 in \cite{Kato} for $n=3$. The same arguments also apply for $n=2$ (see Step 1 in the proof of 
Theorem 1 in \cite{Cuenin-Siedentop}).  Thus 
\begin{align*}
 \|V^{1/2}D_{n}(0)^{-1}V^{1/2}\|\leq \|V^{1/2}r^{1/2}\|^2 \cdot
 \|r^{-1/2}D_{n}(0)^{-1}r^{-1/2}\|<1.
\end{align*}
Hence $1+V^{1/2}D_n(0)^{-1}V^{1/2}$ has a bounded inverse by the Neumann series. \\
Now the claim follows from Theorem 2.2 in \cite{Nenciu} with $A:=D_n(0)$ and $t:=0$.
\end{proof}
Since the assumptions \emph{(i)} and \emph{(ii)} of Theorem 1 in \cite{Morozov-Mueller} are obviously fulfilled, it remains to check assumption \emph{(iii)}. Thus it is enough to find an operator
$L_n:\mathsf{H}^{1/2}(\mathbb{R}^n;\mathbb{C}^{n-1}) \rightarrow \mathsf{H}^{1/2}(\mathbb{R}^n;\mathbb{C}^{n-1})$ such that 
\begin{align*}
               \inf\limits_{\varphi \in \mathsf{H}^{1/2}(\mathbb{R}^n;\mathbb{C}^{n-1})\setminus\{0\} } \frac{\mathtt{d}_{n}\big[\binom{\varphi}{L_n\varphi}\big]+\mathtt{v}\big[\binom{\varphi}{L_n\varphi}\big]}{\big\|\binom{\varphi}{L_n\varphi}\big\|^2}>-1.
\end{align*}
Now we give in three steps an explicit construction of $L_n$ and show that $L_n$ satisfies the requirements.
For $k\in \mathfrak{T}_2$ and $(l,m,s)\in \mathfrak{T}_3$ we define in the first step various constants: 
\begin{align} \label{Definition_c_n}
c_n&:=2(4-n)\frac{\Gamma(\frac{n+1}{4})^2}{\Gamma (\frac{n-1}{4})^2}; \\
c_{2,k}&:=\begin{cases} c_2^{-1}\text{ if }k\in\mathfrak{T}_2^{-}, 
           \\ c_2\text{ if }k\in\mathfrak{T}_2^{+};\end{cases} \label{Definition_c_2_k}\\
c_{3,(l,m,s)}&:=c_3^{2s}.           
\end{align}
In the second step we define the operator $R_n$
 \begin{align} \label{Definition_R}
       R_n: &\bigoplus_{j\in \mathfrak{T}_n}\mathsf{L}^2(\mathbb{R}_{+}) \rightarrow \bigoplus_{j\in \mathfrak{T}_n}\mathsf{L}^2(\mathbb{R}_{+});\ \bigoplus_{j\in \mathfrak{T}_n}\psi_{j}\mapsto\bigoplus_{j\in \mathfrak{T}_n} c_{n,j}  \psi_{T_n^{-1} j} .
\end{align}
Finally we define 
\begin{align} \label{Definition_L}
      L_n:=(\mathcal{U}_n\mathcal{F}_n)^{*}R_n(\mathcal{U}_n\mathcal{F}_n).
\end{align}
The desired properties of $L_n$ are proven in the following lemma:
\begin{lem}\label{Lemma_Talman}
Let $\varphi \in \mathsf{H}^{1/2}(\mathbb{R}^n;\mathbb{C}^{n-1})$ then 
$L_n\varphi \in \mathsf{H}^{1/2}(\mathbb{R}^n;\mathbb{C}^{n-1})$ and the following inequality
\begin{align} \label{equation_lemma_talman}
        \frac{c_n^2-1}{c_n^2+1}\bigg\|\binom{\varphi}{L_n \varphi}\bigg\|^2 \leq
        \mathtt{d}_n\bigg[\binom{\varphi}{L_n \varphi}\bigg]- \frac{1}{4-n}\int_{\mathbb{R}^n}
        \frac{1}{|\mathbf{x}|}\bigg|\binom{\varphi(\mathbf{x})}{\big(L_n \varphi\big)(\mathbf{x})}\bigg|^2\mathrm{d}\mathbf{x}
\end{align}
holds.
\end{lem}
 \begin{proof}
We recall that 
\begin{align*}
            \mathsf{H}^{1/2}(\mathbb{R}^n)=\{\psi \in \mathsf{L}^2(\mathbb{R}^n): (1+|\cdot|^2)^{1/4}\mathcal{F}_n\psi 
            \in \mathsf{L}^2(\mathbb{R}^n)\}.
\end{align*}
Thus the unitarity of $\mathcal{U}_n$ implies
\begin{align} \label{H_1_2-character}
            \mathsf{H}^{1/2}(\mathbb{R}^n)=
            \{\psi \in \mathsf{L}^2(\mathbb{R}^n): 
            \bigoplus\limits_{j\in\mathfrak{T}_n}(1+(\cdot)^2)^{1/4} \big(\mathcal{F}_n\psi\big)_{j} 
            \in \bigoplus\limits_{j\in\mathfrak{T}_n}\mathsf{L}^2(\mathbb{R}_{+})\}.
\end{align}
Moreover we observe that the operator $R_n$ is bounded, which together with \eqref{H_1_2-character} and \eqref{Definition_L} implies that $L_n \varphi \in \mathsf{H}^{1/2}(\mathbb{R}^n)$.\\ 
Now we define the quadratic form $\mathtt{p}$ on $\mathsf{L}^2(\mathbb{R}_{+},(1+p^2)^{1/2}\mathrm{d}p)$ by 
\begin{align*}
                         \mathtt{p}[\chi]:=\int\limits_{0}^{\infty} p |\chi(p)|^2\mathrm{d}p.
\end{align*}
For the proof of \eqref{equation_lemma_talman} we recall that the quadratic form \eqref{q_m} satisfy the inequalities
\begin{equation}
\begin{split} \label{inequalities}
            \mathtt{q}_{k+1/2}[\zeta]&\leq \mathtt{q}_{k-1/2}[\zeta];  \\
            \mathtt{q}_{k+1}[\zeta]&\leq \mathtt{q}_{k}[\zeta]; \\
            \mathtt{q}_{0}[\zeta]&\leq c_3^{-1} \mathtt{p}[\zeta], \quad  \mathtt{q}_{1}[\zeta]\leq c_3 \mathtt{p}[\zeta]; \\
            \mathtt{q}_{-1/2}[\zeta]&\leq 2 c_2^{-1} \mathtt{p}[\zeta], \quad  \mathtt{q}_{1/2}[\zeta]\leq 2 c_2 \mathtt{p}[\zeta];
\end{split}
\end{equation}
for $k\in \mathbb{N}_{0}$ and $\zeta \in \mathsf{L}^2(\mathbb{R}_{+},(1+p^2)^{1/2}\mathrm{d}p)$ (see \cite{Bouzouina} and \cite{Evans-Perry-Siedentop}). \\
By Lemma \ref{Lemma_Coulomb_Channel_Decomposition} we obtain
 \begin{align} 
 \begin{split}\label{proof_talman_1}
  &\int_{\mathbb{R}^n} \frac{|\varphi(\mathbf{x})|^2}{|\mathbf{x}|} \mathrm{d}\mathbf{x} =
  \begin{cases} 
  \sum\limits_{k\in \mathfrak{T}_2} \mathtt{q}_{|k|-1/2}\big[(\mathcal{F}_2\varphi )_{k}\big]
  \text{ if }n=2;\\
  \sum\limits_{(l,m,s)\in \mathfrak{T}_3} \mathtt{q}_{l}
  \big[(\mathcal{F}_3\varphi )_{(l,m,s)}\big]  
  \text{ if }n=3;
  \end{cases} 
    \end{split}
  \end{align}
and by  \eqref{Definition_c_2_k} - \eqref{Definition_L} 
  \begin{align}
  \begin{split}\label{proof_talman_2}
  &\int_{\mathbb{R}^n} \frac{|(L_n\varphi)(\mathbf{x})|^2}{|\mathbf{x}|}\mathrm{d}\mathbf{x} 
  \\
  &=
  \begin{cases}  
  \sum\limits_{k\in \mathfrak{T}^{+}_2} c_{2}^{2} 
  \mathtt{q}_{|k|-\frac{1}{2}} \big[(\mathcal{F}_2\varphi)_{k-1}\big]
  +\sum\limits_{k\in\mathfrak{T}^{-}_2}c_{2}^{-2} \mathtt{q}_{|k|-\frac{1}{2}} 
  \big[(\mathcal{F}_2\varphi )_{k-1}\big]\text{ if }n=2;\\
  \sum\limits_{(l,m,\frac{1}{2})\in \mathfrak{T}^{+}_3}\hspace*{-5mm} c_{3}^{2} 
  \mathtt{q}_{l} \big[(\mathcal{F}_3\varphi )_{(l+1,m,-\frac{1}{2})}\big]+
  \sum\limits_{(l,m,-\frac{1}{2})\in \mathfrak{T}^{-}_3} \hspace*{-5mm}
  c_{3}^{-2} \mathtt{q}_{l} \big[(\mathcal{F}_3\varphi )_{(l-1,m,\frac{1}{2})}\big] \text{ if }n=3.
  \end{cases}
  \end{split}
 \end{align}
 Note that $(l,m,s)\in \mathfrak{T}^{-}_3$ implies $l\in\mathbb{N}$. Hence \eqref{inequalities} implies that the right hand sides of \eqref{proof_talman_1} can be estimated by 
  \begin{align} \label{proof_talman_3}
  (4-n)\left(\sum\limits_{j\in \mathfrak{T}_{n}^{+}} c_n^{-1} 
  \mathtt{p}\big[(\mathcal{F}_n\varphi )_{j}\big]+ 
  \sum\limits_{j\in \mathfrak{T}_{n}^{-}} c_n  \mathtt{p}\big[(\mathcal{F}_n\varphi 
  )_{j}\big]\right);
 \end{align} 
 and the right hand side of \eqref{proof_talman_2} by
 \begin{align} \label{proof_talman_4}
   (4-n)\left(\sum\limits_{j\in \mathfrak{T}_{n}^{+}} c_n
  \mathtt{p}\big[(\mathcal{F}_n\varphi )_{T_n^{-1} j}\big]+ 
  \sum\limits_{j\in \mathfrak{T}_{n}^{-}} c_n^{-1}   \mathtt{p}\big[(\mathcal{F}_n\varphi 
  )_{T_n^{-1} j}\big]\right).
 \end{align}
 By $T_n(\mathfrak{T}_{n}^{\pm})=\mathfrak{T}_{n}^{\mp}$ we conclude that \eqref{proof_talman_4} is equal to \eqref{proof_talman_3}. This together with the relation
 \begin{align*}
 (\mathcal{F}_n L_n\varphi)_{T_n j}=c_{n,T_n j}(\mathcal{F}_n\varphi)_j\text{ for all }j\in\mathfrak{T}_n,
 \end{align*}
 implies
 \begin{align}\label{proof_talman_5}\begin{split}
 &\frac{1}{4-n}\int_{\mathbb{R}^n}\frac{1}{|\mathbf{x}|}\bigg|\binom{\varphi(\mathbf{x})}{(L_n \varphi)(\mathbf{x})}\bigg|^2 \mathrm{d}\mathbf{x}\leq  \\
 & \sum\limits_{j\in \mathfrak{T}_{n}}\quad \int\limits_{\mathbb{R}_{+}}  \left\langle 
 \begin{pmatrix}  
 \big(\mathcal{F}_n\varphi\big)_{j}(p) \\ \big(\mathcal{F}_n L_n\varphi\big)_{T_n j}(p) \end{pmatrix},
 \begin{pmatrix}
 0 & p \\ p & 0
 \end{pmatrix}
  \begin{pmatrix} 
 \big(\mathcal{F}_n\varphi\big)_{j}(p) \\ \big(\mathcal{F}_n L_n\varphi\big)_{T_n j}(p)
 \end{pmatrix}
 \right\rangle_{\mathbb{C}^2} \mathrm{d}p.
 \end{split}
 \end{align}
 A straightforward calculation using \eqref{Definition_c_2_k} - \eqref{Definition_L} gives
 \begin{align}\begin{split}\label{proof_talman_6}
 &\bigg\langle\binom{\varphi}{L_n \varphi},\begin{pmatrix} \mathbb{I}_{\mathbb{C}^{n-1}} &  0 \\  0 &  \mp \mathbb{I}_{\mathbb{C}^{n-1}} \end{pmatrix}\binom{\varphi}{
 L_n \varphi}\bigg\rangle\\
 &=\big(1\mp c_n^{-2}\big)\sum\limits_{j\in \mathfrak{T}^{+}_n } \|\big(\mathcal{F}_n
 \varphi\big)_{j}\|^2 +
 \big(1\mp c_n^{2}\big)\sum\limits_{j\in \mathfrak{T}^{-}_n } \|\big(\mathcal{F}_n
 \varphi\big)_{j}\|^2.
 \end{split}
 \end{align}
 By Lemma \ref{Lemma_alpha_dot_impuls} we know that the right hand side of Relation 
 \eqref{proof_talman_5} plus the minus case of the left hand side of \eqref{proof_talman_6} is equal to   $\mathtt{d}_n\big[\binom{\varphi}{L_n \varphi}\big]$. Thus we obtain \eqref{equation_lemma_talman}  by 
 \eqref{proof_talman_5} and \eqref{proof_talman_6}.
 \end{proof}
\section{Proof of Theorem \ref{Theorem_Esteban_Sere}}\label{Section_Esteban-Sere}
We proceed analogously to the proof of Theorem \ref{Theorem-Talman-minimax-principle}. Thus it is enough to find an operator 
$G_n : P_{n}^{+}\mathsf{H}^{1/2}(\mathbb{R}^n;\mathbb{C}^{2(n-1)}) \rightarrow P_{n}^{-}\mathsf{H}^{1/2}(\mathbb{R}^n;\mathbb{C}^{2(n-1)})$ such that 
\begin{align} \label{condition_esteban_sere}
               \inf\limits_{\varphi \in P_{n}^{+}\mathsf{H}^{1/2}(\mathbb{R}^n;\mathbb{C}^{2(n-1)})\setminus\{0\} } \frac{\mathtt{d}_{n}\big[\varphi+G_n\varphi\big]+\mathtt{v}\big[\varphi+G_n\varphi\big]}{\big\|\varphi+G_n\varphi\big\|^2}>-1
\end{align} 
holds. In the following lemma we prove that a possible choice of $G_n$ is 
\begin{align} \label{Definition_G}
      G_n:=(\mathcal{W}_n\mathcal{F}_n)^{*}E_n(\mathcal{W}_n\mathcal{F}_n),
\end{align}
with
\begin{align}
 \label{Definition_E}
     E_n: \bigoplus_{j\in \mathfrak{T}_n}
           \mathsf{L}^2(\mathbb{R}_{+};\mathbb{C}^2) &\rightarrow \bigoplus_{j \in \mathfrak{T}_n}\mathsf{L}^2(\mathbb{R}_{+};
           \mathbb{C}^2);\\
  \bigoplus_{j\in \mathfrak{T}_n}\Psi_{j}& \mapsto
         \bigoplus_{j\in \mathfrak{T}_n}\frac{1-c_{n,j}
         (\cdot)+\sqrt{1+(\cdot)^2}}
         {c_{n,j}+(\cdot)+c_{n,j}\sqrt{1+(\cdot)^2}}\begin{pmatrix}
         0 & -1 \\ 1 & 0 
         \end{pmatrix}\Psi_{j}.      
\end{align}
\begin{lem}
\label{Lemma_Esteban_Sere}
Let $\varphi \in P_{n}^{+} \mathsf{H}^{1/2}(\mathbb{R}^n;\mathbb{C}^{2(n-1)})$ then $G_n\varphi \in P_{n}^{-} \mathsf{H}^{1/2}(\mathbb{R}^n;\mathbb{C}^{2(n-1)})$ and the relation
\begin{align}\label{Esteban_Sere_Talman}
               L_n(\varphi+G_n\varphi)_1=(\varphi+G_n\varphi)_2
\end{align}
holds. 
\end{lem} 
\begin{rem}
By Lemma \ref{Lemma_Talman} and Relation \eqref{Esteban_Sere_Talman} we conclude 
\eqref{condition_esteban_sere}.
\end{rem}
\begin{proof}[Proof of Lemma \ref{Lemma_Esteban_Sere}]
By Lemma \ref{Lemma_alpha_dot_impuls} we deduce that $\psi \in P^{\pm}_{n}\mathsf{H}^{1/2}(\mathbb{R}^n;\mathbb{C}^{2(n-1)})$ if and only if there exists $\bigoplus\limits_{j\in\mathfrak{T}_n} \zeta_{j}
 \in  \bigoplus\limits_{j\in\mathfrak{T}_n} \mathsf{L}^2(\mathbb{R}_{+};(1+p^2)^{1/2}\mathrm{d}p)$ such that 
\begin{align}\label{spectral_subspace}
      \big(\mathcal{W}_{n} \mathcal{F}_{n}\psi\big)_{j}(p)=
      \begin{cases}\zeta_{j}(p)
      \begin{pmatrix} 1 \\
                      \frac{p}{1+\sqrt{1+p^2}}
      \end{pmatrix} \text{ (''+'' case)};\\
      \zeta_{j}(p)     
      \begin{pmatrix} 
      \frac{-p}{1+\sqrt{1+p^2}} \\
      1
      \end{pmatrix} \text{ (''-'' case)};
      \end{cases}
\end{align}
holds for every $j\in\mathfrak{T}_n$ and $p\in\mathbb{R}_{+}$. Hence we get $G_n\varphi \in P_{n}^{-} \mathsf{H}^{1/2}(\mathbb{R}^n;\mathbb{C}^{2(n-1)})$.\\
By  \eqref{spectral_subspace},\eqref{Definition_E} we obtain that there exists $\bigoplus\limits_{j\in\mathfrak{T}_n}\chi_{j} \in  \bigoplus\limits_{j\in\mathfrak{T}_n} \mathsf{L}^2(\mathbb{R}_{+};(1+p^2)^{1/2}\mathrm{d}p)$ such that
\begin{align*}
      \big(\mathcal{W}_{n} \mathcal{F}_{n}\varphi\big)_{j}(p)=\chi_{j}(p)
      \begin{pmatrix} 1 \\
                      \frac{p}{1+\sqrt{1+p^2}}
      \end{pmatrix} 
\end{align*}
and
\begin{equation}\label{Derivation_ELT}
\begin{split}
&\big((\mathbb{I}+E_n)\mathcal{W}_n\mathcal{F}_n\varphi\big)_{j}=\begin{pmatrix}\tilde{\chi}_{j}\\
c_{n,T_n j}\tilde{\chi}_{j}
\end{pmatrix}\text{ with }\\
&\tilde{\chi}_{j}(p):=\frac{c_{n,j}\Big(p^2+(1+\sqrt{1+p^2})^2\Big)}{(1+\sqrt{1+p^2})(c_{n,j}+p+c_{n,j}\sqrt{1+p^2})}\chi_{j}(p)\text{ for }p\in \mathbb{R}_{+},
\end{split}
\end{equation}
hold  for every $j\in\mathfrak{T}_n$. Hence we get by \eqref{Definition_G},\eqref{Definition_R} and \eqref{Definition_L} the 
relation 
\begin{align*}
     \varphi + G_n \varphi 
     =(\mathcal{W}_n\mathcal{F}_n)^{*}\bigoplus\limits_{j\in\mathfrak{T}_n}
     \begin{pmatrix}\tilde{\chi}_{j}\\
     c_{n,T_n j}\tilde{\chi}_{j}
     \end{pmatrix}
     =\begin{pmatrix}
     (\mathcal{U}_n\mathcal{F}_n)^{*} \bigoplus\limits_{j\in\mathfrak{T}_n}\tilde{\chi}_{j} \\
     L_n (\mathcal{U}_n\mathcal{F}_n)^{*} \bigoplus\limits_{j\in\mathfrak{T}_n}\tilde{\chi}_{j}
     \end{pmatrix}.
\end{align*}
Thus we have proven Relation \eqref{Esteban_Sere_Talman}.
\end{proof}
\section{Proof of Theorem \ref{Theorem-Hardy-Dirac}}\label{Section_Hardy_Dirac}
Since the right hand side of \eqref{Hardy_Dirac_Inequality} is continuous in the $\mathsf{H}^{1}(\mathbb{R}^n;\mathbb{C}^{n-1})$ norm (see Theorem 2.5 in \cite{Herbst}), we can assume that 
 $\varphi \in \mathsf{C}_{0}^{\infty}(\mathbb{R}^n\setminus \{0\};\mathbb{C}^{n-1})\setminus \{0\}$ by the density of  $\mathsf{C}_{0}^{\infty}(\mathbb{R}^n\setminus \{0\};\mathbb{C}^{n-1})$ in $ \mathsf{H}^{1}(\mathbb{R}^n;\mathbb{C}^{n-1})$.\\
By the application of Theorem \ref{Theorem-Talman-minimax-principle}  we obtain
\begin{align} 
 \lambda(v)&\leq \sup\limits_{\psi \in \mathsf{H}^{1}(\mathbb{R}^n\setminus \{0\};\mathbb{C}^{n-1})}I_{n,v,\varphi}(\psi)\text{ with } \label{Talman_Hardy_Dirac1}\\
 I_{n,v,\varphi}&:\mathsf{H}^{1}(\mathbb{R}^n\setminus \{0\};\mathbb{C}^{n-1}) \rightarrow \mathbb{R};\\
 I_{n,v,\varphi}(\psi)&:=
 \frac{\bigg\langle\begin{pmatrix}\varphi \\ \psi\end{pmatrix},\begin{pmatrix} (1+v)\otimes \mathbb{I}_{\mathbb{C}^{n-1}} &  K_n \\  K_n  & (-1+v)\otimes \mathbb{I}_{\mathbb{C}^{n-1}}\end{pmatrix}\begin{pmatrix}\varphi \\ \psi\end{pmatrix}\bigg\rangle}{\bigg\|\begin{pmatrix}\varphi \\ \psi\end{pmatrix}\bigg\|^2} .
 \label{Talman_Hardy_Dirac2}
 \end{align}
 Note that we calculate the suprema in \eqref{Talman_Hardy_Dirac1} over $\mathsf{H}^{1}(\mathbb{R}^n\setminus \{0\};\mathbb{C}^{n-1})$ instead of $\mathsf{H}^{1/2}(\mathbb{R}^n;\mathbb{C}^{n-1})$. 
 This is justified by a density argument, which makes use of the form boundedness of $v\otimes \mathbb{I}_{\mathbb{C}^{2(n-1)}}$ with respect to $D_n(0)$ (see Lemma \ref{Lemma_form_perturbation}) and the density of $\mathsf{H}^{1}(\mathbb{R}^n\setminus \{0\};\mathbb{C}^{n-1})$ in $\mathsf{H}^{1/2}(\mathbb{R}^n;\mathbb{C}^{n-1})$.\\
 Thus the proof of Theorem \ref{Theorem-Hardy-Dirac} basically follows from the following lemma.
\begin{lem} \label{Lemma_I_sup_condition}
We define
\begin{align*}
J_{n,v,\varphi}:(-1,\infty) &\rightarrow \mathbb{R};\\
 J_{n,v,\varphi}(\lambda)&:=
\int\limits_{\mathbb{R}^n}\left(\frac{|K_n\varphi(\mathbf{x})|^2}{1+\lambda -v(\mathbf{x})}+\big(1-\lambda+v(\mathbf{x})\big)|\varphi(\mathbf{x})|^2\right)\mathrm{d}\mathbf{x}.
 \end{align*}
For $\lambda \in (-1,\infty)$, $J_{n,v,\varphi}(\lambda)\leq 0$ implies
 \begin{align*} 
 \sup\limits_{\psi \in \mathsf{H}^{1}(\mathbb{R}^n\setminus \{0\};\mathbb{C}^{n-1})}I_{n,v,\varphi}(\psi)
 \leq \lambda.
 \end{align*}
\end{lem}
\begin{proof}
We introduce 
\begin{align}
\psi_{n,v,\varphi}:(-1,\infty) \rightarrow \mathsf{H}^{1}(\mathbb{R}^n\setminus \{0\};\mathbb{C}^{n-1}); 
\quad \psi_{n,v,\varphi}(\lambda):=\frac{K_n\varphi}{1+\lambda -v}.
\end{align}
For every $\zeta \in \mathsf{H}^{1}(\mathbb{R}^n\setminus \{0\};\mathbb{C}^{n-1})$ the inequality 
\begin{align*}
&\big(I_{n,v,\varphi}\big(\psi_{n,v,\varphi}(\lambda)+\zeta\big)-\lambda\big)
\big(\|\varphi\|^2+\|\psi_{n,v,\varphi}(\lambda)+\zeta\|^2\big)\\
 &=J_{n,v,\varphi}(\lambda)
 +2\Re\langle \zeta,  K_n \varphi-(1+\lambda-v)\psi_{n,v,\varphi}(\lambda)\rangle+ \\
  &\langle K_n\varphi-(1+\lambda-v)\psi_{n,v,\varphi}(\lambda),
 \psi_{n,v,\varphi}(\lambda)\rangle - \langle \zeta, (1+\lambda-v)\zeta\rangle \leq J_{n,v,\varphi}(\lambda)
\end{align*}
holds, and thus we conclude the claim.
\end{proof}
By Lemma \ref{Lemma_I_sup_condition} and \eqref{Talman_Hardy_Dirac1} we obtain 
\begin{align} \label{J_epsilon}
         J_{n,v,\varphi}\big(\lambda(v)-\varepsilon\big)> 0 \text{ for }\varepsilon 
         \in \big(0,1+\lambda(v)\big).
\end{align}
Letting $\varepsilon \searrow 0$ in \eqref{J_epsilon} we obtain Theorem \ref{Theorem-Hardy-Dirac}.
\section{Proof of Theorem \ref{Theorem_operator_core}}\label{Section_Operator_Core}
The proof is based on:
\begin{lem}\label{Lemma_essential_self_adjoint}
Let $\nu\in [0,1/(4-n)]$. The restriction of $\big(\tilde{D}_n(-\nu/ |\cdot|\otimes \mathbb{I}_{\mathbb{C}^{2(n-1)}})\big)^{*}$ 
to $\mathfrak{C}_{n}^{\nu}$ is essentially self-adjoint.
\end{lem}
\begin{proof}
For $m\in\mathfrak{T}_2$ and $(l,m,s)\in\mathfrak{T}_3$ we define 
\begin{align*}
              \kappa_{m}&:=m+1/2;\\
              \kappa_{(l,m,s)}&:=2sl+s+1/2.
\end{align*}
Furthermore we introduce for every $j\in\mathfrak{T}_n$ the operator $D^{j,\nu}$ in $\mathsf{L}^2(\mathbb{R}_{+};\mathbb{C}^2)$ 
by the differential expression 
\begin{align*}
              d^{j,\nu}:=\begin{pmatrix}
              -\frac{\nu}{r} & -\frac{\mathrm{d}}{\mathrm{d}r}-\frac{\kappa_j}{r} \\
              \frac{\mathrm{d}}{\mathrm{d}r}-\frac{\kappa_j}{r} & -\frac{\nu}{r}
              \end{pmatrix}
\end{align*}
on $\mathsf{C}^{\infty}_{0}(\mathbb{R}_{+};\mathbb{C}^2)$. Now we observe that any solution of the equation 
$d^{j,\nu}\varphi=0$ in $\mathbb{R}_{+}$ is a linear combination of the two functions 
\begin{align*}
        \varphi^{\nu}_{j,1}(r):=\begin{cases}
        \begin{pmatrix}
        1 \\ 0
        \end{pmatrix}r^{\kappa_j}\text{ if }\nu=0, \\
        \begin{pmatrix}
        \nu \\ \sqrt{\kappa_j^2 -\nu^2}-\kappa_j
        \end{pmatrix}
        r^{\sqrt{\kappa_j^2 - \nu^2}}\text{ else,}
        \end{cases}
\end{align*}
and 
\begin{align*}
        \varphi^{\nu}_{j,2}(r):=\begin{cases}
        \begin{pmatrix}
        0 \\ 1
        \end{pmatrix}r^{-\kappa_j}\text{ if }\nu=0, \\        
        \begin{pmatrix}
        \nu \\ -\sqrt{\kappa_j^2 -\nu^2}-\kappa_j
        \end{pmatrix}
        r^{-\sqrt{\kappa_j^2 - \nu^2}}\text{ if }0<\nu^2<\kappa_j^2,\\
        \begin{pmatrix}
        \nu\ln(r) \\ 1-\kappa_j \ln(r)
        \end{pmatrix}
        \text{ if }\nu^2=\kappa_j^2.
        \end{cases}
\end{align*}
Through the application of the results of \cite{Weidmann-Oszillationsmethoden} as in Section 2 in \cite{Morozov-Mueller-CLR} we 
obtain that the closure $D^{j,\nu}_{\mathrm{ex}}$ of the restriction of $(D^{j,\nu})^{*}$ to $\mathfrak{C^{j,\nu}}$ is self-adjoint with
\begin{align*}
           \mathfrak{C^{j,\nu}}:=
           \begin{cases}        
           \mathsf{C}^{\infty}_{0}(\mathbb{R}_{+};\mathbb{C}^2)\dot{+}\Span\{\xi \varphi_{j,1}^{\nu}\}\text{ if }
           \kappa_j^2-\nu^2< 1/4;\\ 
           \mathsf{C}^{\infty}_{0}(\mathbb{R}_{+};\mathbb{C}^2)\text{ else}.
           \end{cases}
\end{align*}
Here $\xi$ is a smooth cut-off function with  $\xi \in \mathsf{C}^{\infty}(\mathbb{R}_{+};\mathbb{R}_{+})$, $ \xi(t)=1 $ for $t\in (0,1)$ and $ \xi(t)=0 $ for $t>2$. Thus we conclude the claim by 
\begin{align}\label{cannel_representation_adjoint}
\big(\tilde{D}_n(-\nu/ |\cdot|\otimes \mathbb{I}_{\mathbb{C}^{2(n-1)}})\big)^{*}&=
\left(\mathcal{W}_n \mathcal{M}_n \right)^{*}
\left(\bigoplus\limits_{j\in\mathfrak{T}_n}\left(D^{j,\nu} +\sigma_3\right)^{*}\right)\mathcal{W}_n \mathcal{M}_n\text{ with,}\\
\mathcal{M}_{n}:&=\diag(1,\mathrm{i})\otimes\mathbb{I}_{\mathbb{C}^{n-1}} \nonumber
\end{align}
(see Section 7.3.3 in \cite{Thaller} for $n=2$ and Section 2.1 in \cite{Balinsky-Evans} for $n=3$) and the fact that $\sigma_3$ is 
a bounded operator in $\mathsf{L}^2(\mathbb{R}_{+};\mathbb{C}^2)$.
\end{proof}
\begin{rem}\label{Remark_operator_core_nenciu}
Let $\nu\in\big[0,1/(4-n)\big)$ and $j\in \mathfrak{T}_n$. By the embedding
\begin{align*}
            \mathsf{H}^{1/2}(\mathbb{R}^n)\subset \mathsf{L}^2(\mathbb{R}^n,(1+|\mathbf{x}|^{-1})\mathrm{d}\mathbf{x})
\end{align*}
and \eqref{cannel_representation_adjoint} we obtain that the domain of $\big(\mathcal{W}_n M_n D_n(-\nu/ |\cdot|\otimes \mathbb{I}_{\mathbb{C}^{2(n-1)}})\left(\mathcal{W}_n M_n \right)^{*}\big)_{j}$ is in $\mathsf{L}^2(\mathbb{R}_{+},(1+r^{-1})\mathrm{d}r)$. Hence there is a self-adjoint extension of $D^{j,\nu}$ with domain in $\mathsf{L}^2(\mathbb{R}_{+},(1+r^{-1})\mathrm{d}r)$. By $\xi \varphi_{j,2}^{\nu}\notin \mathsf{L}^2(\mathbb{R}_{+},(1+r^{-1})\mathrm{d}r)$ for $\nu>0$ and Theorem 1.5 in \cite{Weidmann-Oszillationsmethoden} we get that $D^{j,\nu}_{\mathrm{ex}}$ is the unique self-adjoint 
extension of $D^{j,\nu}$ with domain in $\mathsf{L}^2(\mathbb{R}_{+},(1+r^{-1})\mathrm{d}r)$. Therefore, we obtain 
\begin{align*}
               \big(\mathcal{W}_n M_n D_n(-\nu/ |\cdot|\otimes \mathbb{I}_{\mathbb{C}^{2(n-1)}})\left(\mathcal{W}_n M_n \right)^{*}\big)_{j}=D^{j,\nu}_{\mathrm{ex}}.
\end{align*}
We conclude that the closure of $\big(\tilde{D}_n(-\nu/ |\cdot|\otimes \mathbb{I}_{\mathbb{C}^{2(n-1)}})\big)^{*}$ restricted to $\mathfrak{C}_{n}^{\nu}$ is $D_n(-\nu/ |\cdot|\otimes \mathbb{I}_{\mathbb{C}^{2(n-1)}})$.
\end{rem}
As a consequence of Lemma \ref{Lemma_essential_self_adjoint} it remains to prove that $\zeta_{n,m}^{\nu}\in \mathfrak{D}\big(D_n^{\nu}\big)$ for $m\in \{-1/2,1/2\}^{n-1}$ and $(n,\nu)\in \big(\{2\}\times (0,1/2]\big)\cup 
\big(\{3\}\times (\sqrt{3}/2,1]\big) $. We introduce the symmetric and non-negative (by Corollary \ref{Corollary_Dirac_Coulomb_Hardy}) quadratic form $  
\mathtt{q}^{\nu}_{n} $ on  $ \mathsf{C}^{\infty}_{0}(\mathbb{R}^n\setminus\{0\};\mathbb{C}^{n-1}) $ by 
\begin{align*}
          \mathtt{q}^{\nu}_n[\varphi]:=\int\limits_{\mathbb{R}^n}&\bigg(
          \frac{|K_n \varphi|^2}{1+\sqrt{1-\big((4-n)\nu\big)^2}+ 
          \frac{\nu}{|\mathbf{x}|}}+\\
          &\bigg(1-\sqrt{1-\big((4-n)\nu\big)^2}-\frac{\nu}{|\mathbf{x}|}\bigg)
          |\varphi|^2\bigg)\mathrm{d}\mathbf{x}.
\end{align*}
Note that $  \mathtt{q}^{\nu}_n $ is closable by Theorem X.23 in \cite{Simon-Reed-2}. We denote the domain of the closure of $\mathtt{q}^{\nu}_n$ by $\mathfrak{Q}^{\nu}_n$.\\
By the characterisation of $ \mathfrak{D}\big(D_n^{\nu}\big) $ in Theorem 1 in \cite{Esteban-Loss-abstract}, it is enough to show that for all $m\in \{-1/2,1/2\}^{n-1}$ the upper $(n-1)$ spinor of $\zeta_{n,m}^{\nu}$ is in $\mathfrak{Q}^{\nu}_n$, i.e.,
$ \varsigma_{n,m}^{\nu}\in \mathfrak{Q}^{\nu}_n$ with $ \varsigma_{2,m}^{\nu} $ given in polar coordinates by
\begin{align*}
  \varsigma_{2,m}^{\nu}(\rho,\vartheta)&:= \xi(\rho)\rho^{\sqrt{1/4-\nu^2}-1/2} 
            \mathrm{e}^{-\mathrm{i}(m+1/2)\vartheta};
\end{align*}
and $\varsigma_{3,m}^{\nu}$  in spherical coordinates by
\begin{align*}
         \varsigma_{3,m}^{\nu}(r,\theta,\phi):=\xi(r)r^{\sqrt{1-\nu^2}-1}\Omega_{1/2+m_2,m_1,-m_2}(\theta,\phi).
\end{align*}
We achieve this goal by the application of the following abstract lemma
\begin{lem}\label{Lemma_abstract_quadratic_form_domain}
Let $\mathtt{q}$ be a closable and non-negative quadratic form on a dense linear subspace $\mathfrak{Q}$ of the Hilbert space $\mathfrak{H}$ and $\psi\in\mathfrak{H}$. If there is a sequence 
$(\psi_{n})_{n\in\mathbb{N}}\subset\mathfrak{Q}$ with $\sup\limits_{n\in\mathbb{N}}\mathtt{q}[\psi_n]<\infty$ which converges weakly in $\mathfrak{H}$ to $\psi$, then $\psi$ is in the domain of the closure of $\mathtt{q}$.
\end{lem}
\begin{proof}
We denote by $\overline{\mathtt{q}}$ the closure of $\mathtt{q}$ and by $\overline{\mathfrak{Q}}$ the domain of $\overline{\mathtt{q}}$. There is a unique self-adjoint operator $B: \overline{\mathfrak{Q}} \rightarrow \mathfrak{H}$ with 
\begin{align*}
            \overline{\mathtt{q}}[\varphi]= \| B \varphi \|^2\text{ for all }\varphi \in \overline{\mathfrak{Q}}
\end{align*}
by Theorem 2.13 in \cite{Teschl} ($B^2$ corresponds to $A$ there). Thus we know that 
\begin{align*}
         \sup\limits_{n\in\mathbb{N}}\|B \psi_n\|^2<\infty.
\end{align*}
Hence there is a $\Psi\in \mathfrak{H}$ and a subsequence $(B \psi_{n_{k}})_{n_k\in\mathbb{N}}$ of $(B\psi_{n})_{n\in\mathbb{N}}\subset \mathfrak{H}$  that  converges weakly to $\Psi$ by the Banach-Alaoglu Theorem.  This implies that $\big((\psi_{n_k},B\psi_{n_k})\big)_{n_k\in\mathbb{N}}$ converges weakly to $(\psi,\Psi)\in \mathfrak{H}\oplus\mathfrak{H}$. By the closedness of the graph of $B$ and Theorem 8 in 
Chapter 1 of \cite{Cheney} we deduce the claim.
\end{proof}
Now we pick $\upsilon \in \mathsf{C}_{0}^{\infty}(\mathbb{R}_{+})$ such that $\upsilon(r)=\xi(r)$ for all $r\in [1,\infty)$ and $0\leq \upsilon(r)\leq 1$ for $r\in (0,1)$. 
Let $k\in\mathbb{N}$. We define 
\begin{align*}
           \upsilon_{k}(r):=\begin{cases}
                            \upsilon(kr) &\text{ if }r\in(0,1/k]; \\
                             1           &\text{ if }r\in(1/k,1]; \\ 
                             \xi(r)      &\text{ else };
                            \end{cases}                            
\end{align*}
and the function $\varsigma_{2,m,k}^{\nu}$ in the polar coordinates by 
\begin{align*}
  \varsigma_{2,m,k}^{\nu}(\rho,\vartheta)&:= \upsilon_{k}(\rho)
  \rho^{\sqrt{1/4-\nu^2}-1/2} \mathrm{e}^{-\mathrm{i}(m+1/2)\vartheta},
\end{align*}
and $\varsigma_{3,m,k}^{\nu}$  in the spherical coordinates by
\begin{align*}
         \varsigma_{3,m,k}^{\nu}(r,\theta,\phi):=\upsilon_{k}(r)r^{\sqrt{1-\nu^2}-1}
         \Omega_{1/2+m_2,m_1,-m_2}(\theta,\phi).
\end{align*}
The sequence $(\varsigma_{n,m,k}^{\nu})_{k\in\mathbb{N}}$ converges to 
$\varsigma_{n,m}^{\nu}$ in $\mathsf{L}^2(\mathbb{R}^n;\mathbb{C}^{n-1})$. By Lemma 
\ref{Lemma_abstract_quadratic_form_domain} it is thus enough to prove that 
\begin{align} \label{final_need_1}
 \sup\limits_{k\in\mathbb{N}}\mathtt{q}^{\nu}_{n}[\varsigma_{n,m,k}^{\nu}]<\infty.
\end{align}
Let $\varphi \in \mathsf{C}_{0}^{\infty}(\mathbb{R}^n\setminus\{0\};\mathbb{C}^{n-1})$. At first we observe that 
 \begin{align} \label{final_need_2}
            \mathtt{q}^{\nu}_{n} [\varphi]\leq \int_{\mathbb{R}^n}\bigg(\frac{|x|}{\nu} 
            |K_n \varphi|^2
            -\frac{\nu}{|x|}|\varphi|^2+|\varphi|^2\bigg)\mathrm{d}\mathbf{x}.
 \end{align}
A tedious calculation shows
\begin{align} \label{kinetic_angular_momentum_equation}
   K_n=\begin{cases}-\mathrm{i}\mathrm{e}^{i\vartheta}(\partial_\varrho-\frac{1}{\rho}A_2)\text{ with }
  A_2:= - \mathrm{i}\partial_{\vartheta}
   \text{ if }n=2; \\
   -\mathrm{i}\left( \boldsymbol{\sigma} \cdot \frac{x}{\left | x\right |}\right)\left(\partial_r - \frac{1}{r} A_3\right)\text{ with }A_3:=
\boldsymbol{\sigma}\cdot\big(-\mathrm{i}\mathbf{x}
\wedge \nabla\big)\text{ if }n=3.
   \end{cases}
\end{align}
Using \eqref{kinetic_angular_momentum_equation} and integration by parts we obtain  that the right hand side of \eqref{final_need_2} is equal to
\begin{align}\label{final_need_3}
       \int\limits_{\mathbb{R}^n} \left(\frac{|\mathbf{x}|}{\nu} \left | \partial_{|\mathbf{x}|} \varphi
       \right |^2 + \frac{1}{\nu |\mathbf{x}|}\left |(1/(4-n)+A_n)\varphi \right |^2-
       \frac{\bigg(\nu+\frac{1}{(4-n)^2\nu}\bigg)}{|\mathbf{x}|}\left | \varphi \right |^2
       +|\varphi|^2\right)
       \mathrm{d}\mathbf{x}.
\end{align}
By \eqref{final_need_3} and Relation 2.1.37 in \cite{Balinsky-Evans} we obtain
\begin{align}
&\int_{\mathbb{R}^n}\bigg(\frac{|x|}{\nu} 
            |K_n\varsigma_{n,m,k}^{\nu}|^2
            -\frac{\nu}{|x|}| \varsigma_{n,m,k}^{\nu}|^2+| \varsigma_{n,m,k}^{\nu}|^2\bigg)\mathrm{d}\mathbf{x} \nonumber\\ \label{final_need_4}
            =&\int\limits_{0}^{\infty} \bigg(\frac{t^n}{\nu} \left | \partial_{t} 
            \upsilon_{k}(t)t^{\sqrt{(4-n)^{-2}-\nu^2}-(4-n)^{-1}} \right |^2 \\
             &-\nu
            \upsilon_{k}(t)^2t^{2\sqrt{(4-n)^{-2}-\nu^2}-1}
            +\upsilon_{k}(t)^2t^{2\sqrt{(4-n)^{-2}-\nu^2}}\bigg)\mathrm{d}t.
            \nonumber
\end{align}
A straightforward calculation shows that \eqref{final_need_4} is equal to 
\begin{equation}\begin{split} \label{final_need_5}
&\int\limits_{0}^{\infty} \left(\nu^{-1} \upsilon_{k}'(t)^2 t^{2\sqrt{(4-n)^{-2}-\nu^2}+1} +\upsilon_{k}(t)^2 t^{2\sqrt{(4-n)^{-2}-\nu^2}}\right)\mathrm{d}t\\
=&\int\limits_{0}^{1/k} \nu^{-1} k^2 \upsilon'(kt)^2 t^{2\sqrt{(4-n)^{-2}-\nu^2}+1}
\mathrm{d}t
+\int\limits_{1}^{\infty} \nu^{-1}\upsilon'(t)^2 r^{2\sqrt{(4-n)^{-2}-\nu^2}+1}\mathrm{d}t\\
+&\int\limits_{0}^{\infty}  \upsilon_k(t)^2 t^{2\sqrt{(4-n)^{-2}-\nu^2}}\mathrm{d}t.
\end{split}
\end{equation}
An upper bound for the expression in \eqref{final_need_5} is
\begin{align} \label{final_need_6}
\int\limits_{0}^{\infty} \nu^{-1}\upsilon'(t)^2 t^{2\sqrt{(4-n)^{-2}-\nu^2}+1}\mathrm{d}t
+\int\limits_{0}^{\infty} \xi(t)^2 t^{2\sqrt{(4-n)^{-2}-\nu^2}}\mathrm{d}t.
\end{align}
The combination of \eqref{final_need_6}, \eqref{final_need_5} \eqref{final_need_4} and \eqref{final_need_2} implies \eqref{final_need_1}.


\begin{thebibliography}{10}

\bibitem{Balinsky-Evans}
Alexander~A. Balinsky and William~D. Evans.
\newblock {\em {Spectral analysis of relativistic operators}}.
\newblock {I}mperial {C}ollege {P}ress, 2011.

\bibitem{Bouzouina}
Abdelkader Bouzouina.
\newblock {Stability of the two-dimensional {B}rown-{R}avenhall operator}.
\newblock {\em {P}roceedings of the {R}oyal {S}ociety of {E}dinburgh: {S}ection
  {A} {M}athematics}, 132(05):1133--1144, 2002.

\bibitem{Cheney}
Ward Cheney.
\newblock {\em {Analysis for applied mathematics}}, volume 208.
\newblock Springer Science \& Business Media, 2013.

\bibitem{Cuenin-Siedentop}
{J}ean-{C}laude Cuenin and {H}einz {S}iedentop.
\newblock {{D}ipoles in graphene have infinitely many bound states}.
\newblock {\em {J}ournal of {M}athematical {P}hysics}, 55(12), 2014.

\bibitem{Dolbeault-Esteban-Loss-Vega}
Jean Dolbeault, Maria~J. Esteban, Michael Loss, and Luis Vega.
\newblock {An analytical proof of {H}ardy-like inequalities related to the
  {D}irac operator}.
\newblock {\em {J}ournal of {F}unctional {A}nalysis}, 216(1):1--21, 2004.

\bibitem{Dolbeault-Esteban-Sere}
Jean Dolbeault, Maria~J. Esteban, and Eric S{\'e}r{\'e}.
\newblock {On the eigenvalues of operators with gaps. {A}pplication to {D}irac
  operators}.
\newblock {\em {J}ournal of {F}unctional {A}nalysis}, 174(1):208--226, 2000.

\bibitem{Dong-Zong}
Shi-Hai Dong and Zhong-Qi Ma.
\newblock {Exact solutions to the {D}irac equation with a {C}oulomb potential
  in 2+1 dimensions}.
\newblock {\em Physics {L}etters {A}}, 312(1):78--83, 2003.

\bibitem{Esteban-Loss-abstract}
Maria~J. Esteban and Michael Loss.
\newblock {Self-adjointness via partial {H}ardy-like inequalities}.
\newblock In {\em {Mathematical results in quantum mechanics}}, pages 41--47.
  World Sci. Publ., Hackensack, NJ, 2008.

\bibitem{Esteban-Sere}
Maria~J. Esteban and Eric S{\'e}r{\'e}.
\newblock {Existence and multiplicity of solutions for linear and nonlinear
  {D}irac problems}.
\newblock In {\em {{P}artial {D}ifferential {E}quations and their
  {A}pplications}}, volume~12 of {\em {{C}{R}{M} {P}roceedings and {L}ecture
  {N}otes}}, pages 107--118. {A}merican {M}athematical {S}ociety, 1997.

\bibitem{Evans-Perry-Siedentop}
William~D. Evans, Peter Perry, and Heinz Siedentop.
\newblock {The spectrum of relativistic one-electron atoms according to {B}ethe
  and {S}alpeter}.
\newblock {\em {C}ommunications in {M}athematical {P}hysics}, 178(3):733--746,
  1996.

\bibitem{Herbst}
Ira~W. Herbst.
\newblock {Spectral theory of the operator $(p^2+m^2)^{1/2}- {Z}e^2/r$}.
\newblock {\em {C}ommunications in {M}athematical {P}hysics}, 53(3):285--294,
  1977.

\bibitem{Kato}
Tosio Kato.
\newblock {Holomorphic families of {D}irac operators}.
\newblock {\em Mathematische {Z}eitschrift}, 183(3):399--406, 1983.

\bibitem{Morozov-Mueller}
Sergey Morozov and David M{\"u}ller.
\newblock {On the minimax principle for {C}oulomb-{D}irac operators}.
\newblock {\em Mathematische {Z}eitschrift}, 280:733--747, 2015.

\bibitem{Morozov-Mueller-CLR}
Sergey Morozov and David M{\"u}ller.
\newblock {{L}ieb-{T}hirring and {C}wickel-{L}ieb-{R}ozenblum inequalities for
  perturbed graphene with a {C}oulomb impurity}.
\newblock {\em Preprint}, 2016.

\bibitem{Nenciu}
Gheorghe Nenciu.
\newblock {{S}elf-adjointness and invariance of the essential spectrum for
  {D}irac operators defined as quadratic forms}.
\newblock {\em {C}ommunications in {M}athematical {P}hysics}, 48(3):235--247,
  1976.

\bibitem{Simon-Reed-2}
{M}ichael {R}eed and {B}arry {S}imon.
\newblock {\em {{M}ethods of modern mathematical physics II: {F}ourier
  analysis, self-adjointness}}, volume~2.
\newblock {A}cademic {P}ress, 1975.

\bibitem{Talman}
James~D. Talman.
\newblock {Minimax principle for the {D}irac equation}.
\newblock {\em {P}hysical {R}eview {L}etters}, 57(9):1091--1094, 1986.

\bibitem{Teschl}
Gerald Teschl.
\newblock {\em {Mathematical methods in quantum mechanics}}, volume~99.
\newblock American {M}athematical {S}ociety, 2009.

\bibitem{Thaller}
Bernd Thaller.
\newblock {\em {The {D}irac equation}}.
\newblock {S}pringer-{V}erlag, Berlin, 1992.

\bibitem{Weidmann-Oszillationsmethoden}
Joachim Weidmann.
\newblock {Oszillationsmethoden f{\"u}r {S}ysteme gew{\"o}hnlicher
  {D}ifferentialgleichungen}.
\newblock {\em {M}athematische {Z}eitschrift}, 119:349--373, 1971.

\bibitem{Whittaker-Watson}
Edmund~T. Whittaker and George~N. Watson.
\newblock {\em {A course of modern analysis}}.
\newblock {C}ambridge {U}niversity {P}ress, 1996.

\end{thebibliography}
\end{document}